\documentclass[num-refs]{wiley-article}
\papertype{Short Communication}
\usepackage{booktabs}

\usepackage[utf8]{inputenc}
\title{Fast Number Parsing Without Fallback}
\author[1\authfn{1}]{Noble Mushtak}
\author[2\authfn{2}]{Daniel Lemire}
\affil[1]{Northeastern University, Boston, MA, United States}
\affil[2]{DOT-Lab Research Center, Universit\'e du Qu\'ebec (TELUQ), Montreal, Quebec, H2S 3L5, Canada}

\corraddress{Daniel Lemire, DOT-Lab Research Center, Universit\'e du Qu\'ebec (TELUQ), Montreal, Quebec, H2S 3L5, Canada}
\corremail{lemire@gmail.com}
\fundinginfo{Natural Sciences and Engineering Research Council of Canada, Grant Number: RGPIN-2017-03910}

\newcommand{\scaledw}{v} 

\runningauthor{Mushtak and Lemire}
\usepackage[binary-units,per-mode=symbol]{siunitx}
\usepackage{booktabs}
\newcommand{\qedwhite}{\hfill \ensuremath{\Box}}

\usepackage{listings}
\lstdefinestyle{customc}{%
  belowcaptionskip=1\baselineskip,
  breaklines=true,
  xleftmargin=\parindent,
  language=C,
  showstringspaces=false,
  basicstyle=\small\ttfamily,
  keywordstyle=\bfseries\color{green!40!black},
  numberstyle=\tiny,
  commentstyle=\itshape\color{purple!40!black},
  identifierstyle=\bfseries\color{black},
  stringstyle=\color{orange},
   morekeywords={uint64_t,uint32_t,__m256i,__m128i,simd8,uint8_t,UINT64_C},
}
\lstdefinestyle{custompython}{%
  belowcaptionskip=1\baselineskip,
  breaklines=true,
  xleftmargin=\parindent,
  language=Python,
  showstringspaces=false,
  basicstyle=\small\ttfamily,
  keywordstyle=\bfseries\color{green!40!black},
  numberstyle=\tiny,
  commentstyle=\itshape\color{purple!40!black},
  identifierstyle=\bfseries\color{black},
  stringstyle=\color{orange},
   morekeywords={uint64_t,uint32_t,__m256i,__m128i,simd8,uint8_t,UINT64_C},
}
\usepackage{subfloat,subfig}
\usepackage{graphicx}
\usepackage{algorithm}

\usepackage{algpseudocode}
\algdef{SE}[DOWHILE]{Do}{doWhile}{\algorithmicdo}[1]{\algorithmicwhile\ #1}%

\usepackage[table]{xcolor} 
\usepackage{soul}
\usepackage{microtype}

\algnewcommand{\IIf}[1]{\State\algorithmicif\ #1\ \algorithmicthen}
\algnewcommand{\EndIIf}{\unskip\ \algorithmicend\ \algorithmicif}

\usepackage{filecontents}
\usepackage{tikz}
\usepackage{pgfplots,pgfplotstable}

\begin{document}

\maketitle
\begin{abstract}
In recent work, Lemire (2021) presented a fast algorithm to convert 
number strings into binary floating-point numbers. The algorithm has
been adopted by several important systems: e.g., it is part of the
runtime libraries of GCC~12, Rust~1.55, and Go~1.16.
The algorithm parses any number string with a significand containing no more than 19~digits into an IEEE floating-point number. However, there is a check leading to a fallback
function  to ensure correctness. 
This fallback function is never called in practice.
We prove that the fallback is unnecessary.
Thus we can slightly simplify the algorithm and its implementation.


\end{abstract}

\section{Introduction}

Current computers typically support 32-bit and 64-bit  IEEE-754 binary floating-point numbers in hardware~\cite{30711}. Real numbers are approximated by binary floating-point numbers:
a fixed-width integer $m$ (the \emph{significand}) multiplied by 2 raised to an integer exponent $p$: $m \times 2^p$. Numbers are also frequently exchanged as strings in decimal form (e.g., \texttt{3.1416}, \texttt{1.0e10}, \texttt{4E3}). Given  decimal numbers in a string, we must find efficiently the nearest available binary floating-point numbers when loading data from text files (e.g., CSV, XML or JSON documents).

A  64-bit binary floating-point number relies on a 53-bit significand $m$. A 32-bit binary floating-point number has 24-bit significand $m$.
Given a string representing a non-zero number (e.g., \texttt{-3.14E+12}), we represent
it as a sign (e.g., \texttt{-}), an integer $w \in [1,2^{64})$ (e.g., $w=314$)
and a decimal power (e.g., \texttt{10}): $- 314 \times 10^{10}$.
The smallest positive value that can be represented using a 64-bit floating-point number is $2^{-1074}$. We have that $w\times 10 ^{-343} < 2^{-1074}$ for all $w<2^{64}$. Thus if the decimal exponent is smaller than -342, then the number must be considered to be zero.
If the decimal exponent is greater than 308,  the result must be infinite (beyond the range).

In the simplest terms, to represent a decimal number into a 
binary floating-point number, we need to multiply the decimal 
significand by a power of five: 
$m \times 10^q = (m \times 5^q) \times 2^q$. 
Or, conversely, divide it by a power of five. 
For large powers, an exact computation is impractical
thus we use truncated tables.

Lemire's approach to compute the binary significand is given in 
Algorithm~\ref{algo:fancytotalalgo}~\cite{lemire2021number}: the 
complete algorithm needs to compute the binary exponent, handle 
subnormal numbers, and rounding. 
The algorithm effectively multiplies the 64-bit decimal 
significand with a 128-bit truncated power of five, or the 
reciprocal of a power of five, conceptually producing a 192-bit 
product but truncating it to its most significant 128~bits.
When the last 102 bits of the 128-bit product (for 32-bit floating-point numbers) or the last 73 bits of the 128-bit product (for 64-bit floating-point numbers) is made entirely of 
1~bits, a more accurate computation 
(e.g., relying on the full power of five) may produce a 
different result.
Thus the algorithm may sometimes fail and require a fallback 
(line~\ref{line:failure}). 
However, Lemire did not produce an example where a fallback is needed. We want to show that 
it never happens for 32-bit and 64-bit floating-point numbers 
and that the algorithm always succeeds. Hence,
the check and fallback are unnecessary. The check represents a small computational cost. Table~\ref{tab:myperformancetable} presents the number of instructions and cycles per number in one dataset for two systems.\footnote{\url{https://github.com/lemire/simple_fastfloat_benchmark}}
By removing the check, we reduce the number of instructions and CPU cycles per number parsed by 5\% on one system (Intel) and by slightly over 1\% on another (Apple).

\begin{table}[]
    \centering
    \begin{tabular}{ccc}\toprule
    \midrule
                                & Intel Ice Lake, GCC 11 &  Apple M2, LLVM 14 \\
 base instructions per number       & 271  &  299\\
improved instructions per number       & 257 & 295 \\
  CPU cycles per number       & 57.2 & 44.6 \\
improved  CPU cycles per number      & 55.5 & 43.0 \\\bottomrule
    \end{tabular}
    \caption{Performance comparison while parsing the numbers of the canada dataset~\cite{lemire2021number} using CPU performance counters. The reference is the \texttt{fast\_float} library (version 3.2.0). We estimate a 2\% error margin on the cycle/number metric while the instruction count is nearly error-free. We get the improved numbers by removing the unnecessary check.}
    \label{tab:myperformancetable}
\end{table}

\begin{algorithm}
\begin{algorithmic}[1]
\Require  an integer $w \in [1,2^{64})$ and an integer exponent $q\in(-342,308)$
\Require  a table $T$ containing 128-bit reciprocals and  powers of five  for  powers from $-342$ to $308$ \Comment{\cite[Appendix~B]{lemire2021number}}
\State \label{line:norm1}$l \leftarrow$ the number of leading zeros of $w$ as a 64-bit (unsigned) word
\State \label{line:norm2}$\scaledw \leftarrow 2^l \times w$  \Comment{We normalize the significand.}
\State  \label{line:multiplication} Compute the 128-bit truncated product  $z \leftarrow (T[q] \times \scaledw) \div 2^{64}$.
\If{($z \bmod 2^{73} = 2^{73}-1$ for 64-bit numbers or $z \bmod 2^{102} = 2^{102}-1$ for 32-bit numbers) and $q \notin [-27,55]$ \label{line:failure} } 
\State \textbf{Fallback~needed} \Comment{\cite[Remark~1]{lemire2021number}} 
\EndIf{} 
\State \textbf{Return }\label{line:binarysignificand}$m \leftarrow$ the most significant 55~bits (64-bit) or 26~bits (32-bit) of the product $z$
\end{algorithmic}
\caption{%
Algorithm to compute the binary significand from  a  positive 
decimal floating-point number $w \times 10^q$ for the IEEE-754 
standard.
We compute more bits to allow for exact rounding and to account for  a possible leading zero bit. We use the convention that
$a \bmod b$ is the remainder of the integer division of $a$ by $b$.\label{algo:fancytotalalgo}}
\end{algorithm}

\section{Related Work}

Clinger~\cite{10.1145/93548.93557,10.1145/989393.989430} was maybe earliest in describing accurate and efficient decimal-to-binary conversion techniques. He proposed a fast path using the fact that small powers of 10 can be represented exactly as floats. His fast path is still useful today.  Gay~\cite{gay1990correctly} 
implemented  a fast general decimal-to-binary implementation that is still popular. Gay's strategy is to first find quickly a close approximation, and then to refine it with exact big-integer arithmetic.

The reverse problem, binary-to-decimal conversion, has received much attention~\cite{10.1145/989393.989431}. 
Adams~\cite{10.1145/3192366.3192369,10.1145/3360595} bound the maximum and minimum of $ax \bmod b$ over an interval starting at zero, to show that powers of five truncated to 128~bits are sufficient to convert 64-bit binary floating-point numbers into equivalent decimal numbers.

\section{Continued Fractions}
\label{sec:continued}
Given a sequence of integers $a_0, a_1, \ldots$, 
the expression
$
    a_0 + \frac{1}{a_1 + \frac{1}{a_2 + \frac{1}{\cdots}}}
$
is a  \emph{continued fraction}. 
 We can write a rational number $n/d$ where $n$ and $d>0$ are integers as a continued fraction by computing  $a_0, a_1, \ldots$.
 A continued fraction  represents 
a sequence of converging values  $ c_0 = a_0 , 
    c_1 = a_0 + \frac{1}{a_1}, 
    c_2 = a_0 + \frac{1}{a_1+ \frac{1}{a_2}}, 
    c_3 = a_0 + \frac{1}{a_1+ \frac{1}{a_2+  \frac{1}{a_3}}}, \ldots$
We call each such value ($c_0, c_1, c_2, \ldots$) a 
\emph{convergent}.
Each convergent is a rational number: we can 
find integers $p_n$ and $q_n$ such that $c_n = p_n/q_n$. 
 We are interested in computing the convergents quickly.
We can check that $p_0 = a_0, q_0 = 1$, $p_1 = a_1  p_0 + 1, 
q_1 = a_1 q_0$, and
 $p_2 = a_2  p_1 + p_0, 
q_2 = a_2  q_1 + q_0$.
We have the general formulas 
 $p_n = a_n  p_{n-1} + p_{n-2}, 
q_n= a_n  q_{n-1} + q_{n-2}$. 
These recursion formulas were known to Euler. 
Rational numbers that are close to a value $x$ are also 
convergents according to the following theorem due to 
Legendre~\cite{hardy1979introduction,legendre1808essai}.

\begin{theorem}\label{theorem:legendre}
For $q>0$ and $p,q$ integers, if the fraction $p/q$ is such that
$\vert p/q-x\vert < \frac{1}{2q^2}$, then 
 $p/q$ is a convergent of $x$.
\end{theorem}


\section{Ruling Out Fallback Cases}


Let $w$ and $q$ be the decimal significand and the decimal exponent, respectively, of the decimal number to be converted into a binary floating-point number. We require that the decimal significand
is a positive 64-bit number. We multiply $w$ by a power of two to produce $\scaledw \in[2^{63}, 2^{64})$. 
Let $T[q]$ be the 128-bit number representing a truncated power of five or the reciprocal of a power of five.  The fallback condition may be needed when the least significant 73~bits of the most significant 128~bits of the product of $\scaledw$ and $T[q]$ are all 1~bits. When parsing 32-bit floating-point numbers, we have a
more generous margin (102~bits instead of 73~bits), so it is sufficient to 
review the 64-bit case.

Assume that the least significant 73~bits of the most significant 128~bits of the product of $\scaledw$ and $T[q]$ are all 1~bits. It is equivalent to requiring that the number made of the least significant 137~bits of the product, 
$r=(T[q]\times \scaledw)\bmod 2^{137}$, is  larger than or equal to $2^{137}-2^{64}$: $r \geq 2^{137}-2^{64}$. By the next lemma, we have that $\scaledw$ is the denominator of a convergent of $T[q]/2^{137}$.

\begin{lemma}\label{lemma:convergent}
Given positive integers $T[q]<2^{128}$ and $\scaledw<2^{64}$, we have that $(T[q]\times \scaledw)\bmod 2^{137}\geq 2^{137}-2^{64}$ implies that $n/\scaledw$, for 
$n  = 1+\lfloor T[q] \times \scaledw /2^{137} \rfloor$, is a convergent of $T[q]/2^{137}$.
\end{lemma}
\begin{proof}
 By definition,  we have that
 $a = b \lfloor a /b \rfloor + (a \bmod b)$ given two integers $a,b$ when $b>0$.
 Thus, 
given $r=(T[q]\times \scaledw)\bmod 2^{137}$, we have
that $(T[q]\times \scaledw) = 2^{137}\lfloor T[q] \times \scaledw /2^{137} \rfloor + r.$
Substracting $r$ from both sides, we get
$(T[q]\times \scaledw) -r  = 2^{137}\lfloor T[q] \times \scaledw /2^{137} \rfloor.$
Adding $2^{137}$ on both sides, we get 
$(T[q]\times \scaledw) + 2^{137}-r  = 2^{137} + 2^{137}\lfloor T[q] \times \scaledw /2^{137} \rfloor = 2^{137} n.$
Let $x=2^{137}-r$, then we have
 $(T[q]\times \scaledw) + x =  2^{137} n$. 
 And because $\scaledw$ is non-zero we can divide by $ \scaledw \times 2^{137}$  throughout to get 
  $\frac{T[q]}{2^{137}}  + \frac{x}{\scaledw \times 2^{137}} =  n/\scaledw. $

Because $(T[q]\times \scaledw)\bmod 2^{137}\geq 2^{137}-2^{64}$,
we have that $r\geq 2^{137}-2^{64}$, and so
$2^{137}-r \leq 2^{64}$ or $x\leq 2^{64}$.
  
  Finally, we have
\begin{flalign*}
&& \left\vert \frac{ T[q]}{2^{137}}  - \frac{n}{\scaledw} \right \vert & = \frac{x}{ \scaledw\times 2^{137}}  && \text{} \\
\Rightarrow &&  &\leq \frac{1}{  \scaledw\times 2^{73}} && \text{because $ x\leq 2^{64}$} \\
\Rightarrow &&  &<\frac{1}{  2 \scaledw^2}  && \text{because $\scaledw<2^{64}$.} 
\end{flalign*}
Hence,  we have that $n/\scaledw$ is a convergent of  
the rational number $T[q]/2^{137}$ by Theorem~\ref{theorem:legendre}. \qedwhite
\end{proof}

Thus, if the fallback condition is triggered by some decimal significand $\scaledw$ and some decimal scale $q$, the significand $\scaledw$ must be the denominator of a convergent of $T[q]/2^{137}$.

It remains to show that it suffices to check the convergent as simple fractions.
If  $n/\scaledw$ is a convergent, then $(n/d)/(\scaledw/d)$ where $d= \gcd(n, d)$ is the simple counterpart. If $\scaledw<2^{64}$, then $\scaledw/d<2^{64}$. We want to show that if $(T[q]\times \scaledw)\bmod 2^{137}\geq 2^{137}-2^{64}$
then $(T[q]\times \frac{\scaledw}{d} )\bmod 2^{137}\geq 2^{137}-2^{64}$.

Consider the following technical lemma which relates $ a\times \frac{\scaledw}{d} \bmod 2^{137}$ and $(a\times \scaledw)\bmod 2^{137}$.
\begin{lemma}\label{lemma:rescale}
    Given positive integers  $a$ and $\scaledw$, we have that
    $
     \left (   a\times \frac{ \scaledw}{d} \right ) \bmod 2^{137}
    = 2^{137}+  \frac{(a\times \scaledw)\bmod 2^{137}- 2^{137}}{d}
   $
where $d= \gcd (\scaledw, 1+\lfloor a\times \scaledw / 2^{137}\rfloor)$.
\end{lemma}

\begin{proof}
We have that $a \times \scaledw = 2^{137}( 1+\lfloor a\times \scaledw / 2^{137}\rfloor)  +( (a\times \scaledw)\bmod 2^{137}) - 2^{137}$.
Given $d= \gcd (\scaledw,  1+\lfloor a\times \scaledw / 2^{137}\rfloor)$, then 
$( 1+\lfloor a\times \scaledw / 2^{137}\rfloor)/d$ and $\scaledw/d $ are integers.
We have that $a \times \scaledw/d = 2^{137} ( 1+\lfloor a\times \scaledw / 2^{137}\rfloor)/d +( (a\times \scaledw)\bmod 2^{137} - 2^{137})/d$.
Thus it follows that $(a\times \scaledw/d)\bmod 2^{137} = 2^{137}+( (a\times \scaledw)\bmod 2^{137} - 2^{137})/d$.\qedwhite
\end{proof}

Suppose that  $(T[q]\times \scaledw)\bmod 2^{137}\geq 2^{137}-2^{64}$, then by Lemma~\ref{lemma:rescale}, we have 
\begin{align*}
\left (  T[q]\times \frac{\scaledw}{d} \right )\bmod 2^{137}
    &= 2^{137}+  \frac{(T[q]\times \scaledw)\bmod 2^{137}- 2^{137}}{d}  \\
 &\geq 2^{137}+  \frac{2^{137}-2^{64}- 2^{137}}{d} \\
 &\geq 2^{137}-2^{64}/d.   
\end{align*}
Hence we have that if $(1+\lfloor T[q] \times \scaledw /2^{137} \rfloor)/\scaledw$ is a convergent such that $(T[q]\times \scaledw)\bmod 2^{137}\geq 2^{137}-2^{64}$, then the 
simplified convergent,
$\scaledw' \leftarrow \scaledw/d$, also satisfies
$(T[q]\times \scaledw')\bmod 2^{137}\geq 2^{137}-2^{64}$ since $d\geq 1$. Thus it is sufficient to check the convergents in simplified form.
We have proven the following proposition.

\begin{proposition}\label{proposition:final}
Given a positive integer $T[q]$, we have that  there exists a positive integer
$\scaledw<2^{64}$ such that $(T[q]\times \scaledw)\bmod 2^{137}\geq 2^{137}-2^{64}$
if and only if there is a continued-fraction convergent $p'/q'$ of  $T[q]/2^{137}$
such that $p'/q'$ is a simple fraction ($\gcd(p',q')=1$), $q'<2^{64}$
and  $(T[q]\times q')\bmod 2^{137}\geq 2^{137}-2^{64}.$
\end{proposition}

For any decimal exponent $q$, to check whether there exists a decimal significand $1 \leq \scaledw < 2^{64}$ which triggers the fallback condition, it suffices
to examine all  convergents of $T[q]/2^{137}$ with a denominator less than $2^{64}$ and check whether  $(T[q]\times \scaledw)\bmod 2^{137}\geq 2^{137}-2^{64}$.
We can convert  $T[q]/2^{137}$ into a continued fraction and compute the coefficients
$a_0, a_1, a_2, \ldots$. We can then use the recursive formulas for computing convergents to enumerate all convergents as simple fractions~\cite{hardy1979introduction}.
Given convergents as simple fraction $p_0/q_0, p_1/q_1, \ldots$, we check whether $(T[q] \times q_i)\bmod 2^{137}\geq 2^{137}-2^{64}$ for $i =0,1,\ldots$ It is enough to check the convergents with a denominator smaller than $2^{64}$.
We implement this algorithm with a Python script: the script reports that no fallback is needed.
Hence  no fallback is required for 64-bit significands in the number-parsing algorithm described by Lemire~\cite{lemire2021number}.

\bibliography{fastfloat}

\end{document}